\documentclass[twoside,leqno,twocolumn]{article}

\usepackage[letterpaper]{geometry}

\usepackage{mymacros} 

\usepackage{xr}
\begin{document}


\title{{\Large Maximizing diversity over clustered data}\thanks{This
work was supported
by the Academy of Finland project 
``Active knowledge discovery in graphs (AGRA)'' (313927), 
the EC H2020 RIA project ``SoBigData++'' (871042), 
and the Wallenberg AI, Autonomous Systems and Software Program (WASP).}}
\author{Guangyi Zhang\thanks{Department of Computer Science, KTH Royal Institute
of Technology, Sweden. This work was done while the author was with Aalto University.}
\and Aristides Gionis\thanks{Department of Computer Science, KTH Royal Institute
of Technology, Sweden, and Department of Computer Science, Aalto University, Finland.}}

\date{}

\maketitle


\fancyfoot[R]{\scriptsize{Copyright \textcopyright\ 2020 by SIAM\\
Unauthorized reproduction of this article is prohibited}}





\begin{abstract} 
Maximum diversity aims at selecting a diverse set of high-quality objects from a collection, which is a fundamental problem and has a wide range of applications, e.g., in Web search.
Diversity under a uniform or partition matroid constraint naturally describes useful cardinality or budget requirements, and admits simple approximation algorithms~\cite{borodin2012max}.
When applied to clustered data, however, popular algorithms such as picking objects iteratively and performing local search lose their approximation guarantees towards maximum intra-cluster diversity because they fail to optimize the objective in a global manner.
We propose an algorithm that greedily adds a pair of objects instead of a singleton, 
and which attains a constant approximation factor over clustered data.
We further extend the algorithm to the case of monotone and submodular quality function, and under a partition matroid constraint.
We also devise a technique to make our algorithm scalable, and on the way we obtain a modification that gives better solutions in practice while maintaining the approximation guarantee in theory. 
Our algorithm achieves excellent performance, compared to strong baselines in a mix of synthetic and real-world datasets.%
\end{abstract}

\section{Introduction}\label{sec:intro}

The {\em dispersion problem} aims at selecting a diverse set of objects from a collection. 
The problem has been studied extensively for 
different notions of dispersion and different settings. 
In their seminal work, Ravi et al.~\cite{ravi1994heuristic} 
studied the problem of selecting a subset of objects in a metric space
so as to maximize some dispersion criterion,  
such as the minimum distance among all pairs of selected objects (\textsc{MaxMin} or {\em remote edge}), 
and the sum of all pairwise distances (\textsc{MaxSum} or {\em remote clique}), 
and provided algorithms with constant-factor approximation guarantees.
Stronger theoretical guarantees have been established later 
for the metric case~\cite{birnbaum2009improved,hassin1997approximation}, 
while a polynomial-time approximation scheme is possible 
for negative-type distances~\cite{manzano2016approximation}.

The \textit{diversification problem} is an extension of the dispersion problem
that aims to maximize not only the pairwise distances of the selected objects, 
but also quality (or relevance in information retrieval) over those selected objects,
via a weighted sum of both.
The problem has important applications in 
databases, operations research, and information retrieval.
Gollapudi and Sharma \cite{gollapudi2009axiomatic} showed 
that if the quality function is modular, 
the diversification problem can be re-casted as a dispersion problem. 
Borodin et al.~\cite{borodin2012max} further extend the diversification problem 
to permit a submodular quality function, which is common in real-life applications. 
Furthermore, they generalized the diversification setting to allow for an arbitrary matroid constraint. 
The algorithms proposed by Borodin et al.\ rely on 
a greedy strategy that selects one object at a time, 
as well as a local-search strategy. 

The aforementioned works, however, assume that the objects to be selected come from a single collection.
On the other hand, in many cases, the collection of objects is split 
into a number of potentially-overlapping clusters.
Consider the following example on analyzing communities of different political leaning in the Twitter network.
Users from different communities discussing a topic may use the same hashtags, thus, 
inducing an overlapping clustering on the set of hashtags.
We want to design a \emph{description} functionality for each community,
which creates descriptors by selecting a subset of hashtags used by users in each community.
We are naturally interested in selecting a high-quality and diverse subset of hashtags, 
and at the same time, increasing exposure of different hashtags as much as possible, 
i.e., a hashtag will not be repeatedly selected.

In this paper we focus on the problem of maximizing diversity over data that form overlapping clusters, 
which models the setting discussed in the previous example.
In more detail, 
we maximize the sum of dispersion within each selected set.
The setting can be seen as applying {\sc MaxSum} diversification 
to each cluster separately under a global partition matroid constraint.
Similar to Borodin et al.,
we assume that distances between our objects are measured by a metric, 
while the quality function is monotone non-decreasing and submodular.
To the best of our knowledge, Abbassi et al.~\cite{abbassi2013diversity} are the only work that covers diversity maximization over clustered data, but they focus on global dispersion, 
which is a sum of intra- and inter-cluster dispersion and may cause undesirable consequences in many applications.
For example, in the above example, the presence of inter-cluster dispersion prevents one community from selecting hashtags similar to the one taken by the other community, and leads to incomplete summarization of communities.

Our first observation is that in the presence of overlapping clusters, 
the approximation algorithms of Borodin et al.\ 
(greedy and local search)
have unbounded approximation ratios.
We then propose algorithms with provable approximation guarantees.
Our key finding is that it is beneficial to greedily 
add \emph{pairs of objects} in the solution, instead of singleton objects.
This idea is inspired by the work of Hassin et al.~\cite{hassin1997approximation}.

Additionally, we propose a modification of our algo\-rithm, 
which improves scalability, 
and also enables additional flexibility. 
The modified algorithm allows to obtain better solutions in practice
while maintaining the approximation guarantee in theory.

We summarize our contributions as follows.
\squishlist
\item
To the best of our knowledge, 
this is the first work to study the 
problem of intra-cluster diversity maximization in the presence of overlapping clusters. 
\item
We propose a greedy algorithm
with constant-factor approximation guarantee, 
while other popular heuristics are shown to have unbounded approximation ratio.
We further extend the analysis to the case of monotone and submodular quality function, 
and under a partition matroid constraint.
\item 
We show how to make our algorithm scalable, 
and on the way we obtain a modification that
gives better solutions in practice 
while maintaining the approximation guarantee in theory. 
\item Our algorithm achieves excellent performance, 
compared to strong baselines in a mix of synthetic and real-world datasets. 
\squishend

The rest of the paper is organized as follows.
We first discuss related work in Section~\ref{sec:related}.
Then a formal definition for the problem is formulated in Section~\ref{sec:def}, 
followed by theoretical analysis in Sections \ref{sec:analysis} and~\ref{sec:subm}, 
with and without a quality function, respectively.
A modifi\-cation is introduced in Section~\ref{sec:scalability} to cope with large-scale data.
We evaluate our algorithms on both synthetic and real-world datasets in Section~\ref{sec:experiment}.
Finally, we conclude in Section~\ref{sec:conclusion}.

For lack of space, 
the proofs of most of our technical results, 
together with the case of partition matroid constraint, 
and some additional experiments are given in the Appendix.

\section{Related work}\label{sec:related}
Maximum dispersion was first studied by Ravi et al.~\cite{ravi1994heuristic}.
Gollapudi and Sharma \cite{gollapudi2009axiomatic} incorporate a quality objective into the dispersion framework.
Borodin et al.~\cite{borodin2012max} extend the dispersion problem to submodular quality functions.
Abbassi et al.~\cite{abbassi2013diversity} study global dispersion under a partition matroid constraint over clustered data, whereas we study a sum of intra-cluster dispersion.

Most of the existing work on maximum dispersion, including the approaches mentioned above, 
adopts a greedy strategy of 
adding iteratively a single element or local-search strategy.
One exception is the work of Cevallos et al.~\cite{manzano2016approximation}, 
who apply convex quadratic programming with a negative-type distance function.
Bhaskara et al.~\cite{bhaskara2016linear} solve a similar \textsc{Sum-Min} problem using an 
Linear Programming relaxation.
Our algorithm is inspired by the work of Hassin et al.~\cite{hassin1997approximation}, 
who utilize maximum matching for multi-subset selections
and a more efficient greedy algorithm on disjoint edges for the case of a single subset.

In web search, 
the importance of providing users a list of documents that are not only relevant to their queries
but also diverse has been noticed early on.
Carbonell et al.~\cite{carbonell1998use} propose a greedy algorithm with respect to {\em Maximal Marginal Relevance} 
(MMR) to reduce the redundancy among returned documents.
Zhai et al.~\cite{zhai2015beyond} invent two novel distances for documents, and evaluate diversity using coverage on subtopics.
Agrawal et al.~\cite{agrawal2009diversifying} incorporate an existing taxonomy of topics and develop a probabilistic model whose goal is to cover every intention of a query, where the probabilities of documents covering an intention are considered independent, determined by their relevance to the query and their topic distributions.

Zhang et al.~\cite{zhang2005improving} form an affinity graph of documents, where the most informative document given by PageRank is selected at each iteration, and then its neighbors are immediately removed to avoid redundancy.
Radlinski et al.~\cite{radlinski2008learning} assume that users with different interpretations of the same query click diverse documents, leading to the idea of optimizing a ranking of documents using only user clickthrough data such that the probability of users not clicking any document on the list is minimized.
Compared to these approaches, diversification provides a solid theoretical framework, 
which can be applied to many different settings, 
once an appropriate object representation and distance metric are defined.

Other related work includes the team-formation problem \cite{lappas2009finding}, 
where it is asked to select team members who bear a small communication cost, 
formulated as a Steiner tree or diameter.
Esfandiari et al.~\cite{esfandiari2019optimizing} have a different focus, 
partitioning students into different groups to encourage peer learning, 
so as to optimize some special affinity structure such as a star in each~group.

The scalability of dispersion algorithms becomes an essential concern as the data size increases.
Distributed algorithms have been invented to tackle the computational complexity for dispersion problems \cite{zadeh2017scalable}.
However, little work focuses on the scalability of dispersion algorithms in a clustering setting.

\section{Notation and problem definition}\label{sec:def}

\note[Guangyi]{Lets sort out lower-case indexing letters we are using so far. 
$j$ to index clusters.
$i$ is later used to index the $i$-th pair taken.
Other lower-case letters are: $n,m,k$ for size, $b$ for cardinality ($\pr{b}$ for a rounded-up even number), $p$ for pair, $d$ for distance, and $u,v,w,x,y,l$ for elements.
}
\note[Aris]{perfect!}

In this section we formally define the {\em intra-cluster dispersion} problem, 
which is the focus of our paper. 

We start with a {\em ground set} $\univset =\{\elem_1,\ldots,\elem_{\nunivset}\}$,
and a {\em distance metric} $\dist:\univset\times\univset\rightarrow\mathbb{R}_{+}$ defined over \univset.
In some cases a {\em quality function} $\quality:2^\univset\rightarrow\mathbb{R}$ is available, 
which evaluates the quality of subsets of \univset according to some desirable property.

In addition we assume that we are given as input a 
clustering $\clustering=\{\cls_1,\ldots,\cls_{\ncls}\}$ over the ground set $\univset$, 
which may contain overlapping clusters. 
Finally, we are given cardinality upper bounds $\nsel_1,\ldots,\nsel_{\ncls}$, 
one for each cluster of \clustering, 
i.e., $\nsel_j$ is a cardinality upper bound for cluster $\cls_j$, for $j=1,\ldots,\ncls$.

The objective of the intra-cluster dispersion problem
is to pick pairwise disjoint sets $\sel=\{\sel_1,\ldots,\sel_{\ncls}\}$, 
with $\sel_j\subseteq\cls_j$ and $|\sel_j|\le\nsel_j$, for $j=1,\ldots,\ncls$,
so as to maximize the dispersion function
\begin{equation}
\label{equation:intra-dispersion}
\distintra(\sel) = \sum_{j=1}^{\ncls} \dist(\sel_j),
\end{equation}
\note{Every pair counted twice in $\dist(\sel_j)$.}
where $\dist(\sel_j) = \sum_{(u,v): u,v\in \sel_j} \dist(u,v)$
is the sum of all pairwise distances of the elements $\sel_j$ picked from cluster~$\cls_j$.

When a quality function $\quality$ is available, 
it can be incorporated in the intra-cluster dispersion problem
by asking to maximize the objective function
\begin{equation}
\label{equation:quality-intra-dispersion}
\distintrasubm(\sel) = \quality(\sel) + \lambdasubm\,\distintra(\sel),
\end{equation}
where $\lambdasubm$ is a user-defined parameter that is used to provide a trade-off between quality and dispersion.
Equation~(\ref{equation:quality-intra-dispersion}) is often used to formulate the {\em diversification problems}~\cite{gollapudi2009axiomatic},
where we typically want to select high-quality sets (as measured by \quality)
which are also diverse (as measured by \distintra).

We refer to the problems of maximizing Equations~(\ref{equation:intra-dispersion}) 
and~(\ref{equation:quality-intra-dispersion})
by \icd and \icdq, respectively.
In the case of a single cluster, we refer to them as \scd and \scdq, respectively.

An interesting generalization of the intra-cluster dispersion problem
is when in addition to the clustering $\clustering=\{\cls_1,\ldots,\cls_{\ncls}\}$ 
we are also given a partition  
$\partitionset = \{\partition_1,\ldots,\partition_{\npartition}\}$, 
the goal remains unchanged with an additional requirement of selecting at most one element from each set in the partition~\partitionset. 
Without loss of generality, we can assume that each element forms its own partition, 
resulting in the constraint of the sets in \sel being disjoint. 
However, our methods generalize to a general partition matroid constraint 
by removing a whole set of the partition when we select a single element from it.
A detailed elaboration can be found in Section \sref{subsec:partition}. 

Finally, recall that a set function 
$\quality:2^\univset\rightarrow\mathbb{R}$
is {\em monotone non-decreasing}
if $\quality(B)\ge\quality(A)$, for all $A\subseteq B\subseteq\univset$.
The function \quality is {\em sub\-modular} if it sat\-i\-sfies the {\em ``diminishing returns''} property, 
that is, for all $A,B\subseteq\univset$ with $A\subseteq B$ and for all $\elem\in\univset\setminus B$, 
it holds that
$\quality(A\cup\{\elem\}) - \quality(A) \ge \quality(B\cup\{\elem\}) - \quality(B)$.
A set function \quality is {\em super\-modular} if $-\quality$ is sub\-modular, 
and it is {\em modular} if it is both sub\-modular and super\-modular. 

\spara{Complexity.}
The \icd and \icdq problems studied in this paper
extend the {\em max-sum diversification} problem~\cite{borodin2012max}, 
which aims at maximizing the objective (\ref{equation:quality-intra-dispersion})
in a single cluster. 
The max-sum diversification is proven to be \np-hard~\cite{hansen1995dispersing}.
As a consequence, the problems \icd and \icdq are \np-hard, 
as they generalize the max-sum diversification problem in the case that clusters are available.
Thus, our objective is to develop approximation algorithms with provable guarantees. 

\section{Maximizing dispersion}\label{sec:analysis}

We start our discussion about limitations of existing algorithms for the \icd and \icdq problems 
(i.e.,\ maximizing Equation~(\ref{equation:intra-dispersion}) and Equation~(\ref{equation:quality-intra-dispersion}), respectively), 
and proceed to the proposed algorithm for the \icd problem.
The case of \icdq is discussed in Section~\ref{sec:subm}.

For the \scdq problem, 
which is a special case of the \icdq problem when no clustering is available, 
the {\em greedy algorithm} is a natural choice. 
The greedy adds elements into \sel iteratively, 
each time selecting the element that yields the maximal gain in a slightly-modified 
version of the objective function (\ref{equation:quality-intra-dispersion}).
Borodin et al.~\cite{borodin2012max} showed that if the quality function \quality is submodular, 
the greedy algorithm achieves a factor-2 approximation guarantee.

A natural extension of the greedy algorithm for the \icdq problem is the following: 
iterate over the clusters $\cls_1,\ldots,\cls_{\ncls}$, 
and for each cluster $\cls_j$, 
form the set $\sel_j$ by iteratively adding $\nsel_j$ elements into $\sel_j$ from $\cls_j$; 
each time add the element that yields the maximal gain for the objective function (\ref{equation:quality-intra-dispersion}).
We refer to this simple algorithm as \gelms.

Another powerful technique in combinatorial optimization is {\em local search}:
start by a feasible solution and perform incremental changes until the objective cannot be improved. 
For the \icdq problem, 
an intuitive way to apply the local-search strategy
is by {\em element swapping}, 
which is finding a cluster
for which it is possible to swap an element in the current solution and improve the objective.
We call this algorithm \lsi.
When the objective is replaced with global dispersion, we call it \lsg.

Our first result is to show that for the \icd and \icdq problems
the algorithms \gelms and \lsi do not have a bounded approximation factor. 
This is in contrast to the approximation guarantee of \gelms
for the \scd problem.

\begin{lem}
\label{lemma:unbounded}
The approximation factor of algorithms \gelms and \lsi is unbounded.
\end{lem}
\begin{proof}
\begin{figure}[t] 
\centering
\includegraphics[width=.4\textwidth]{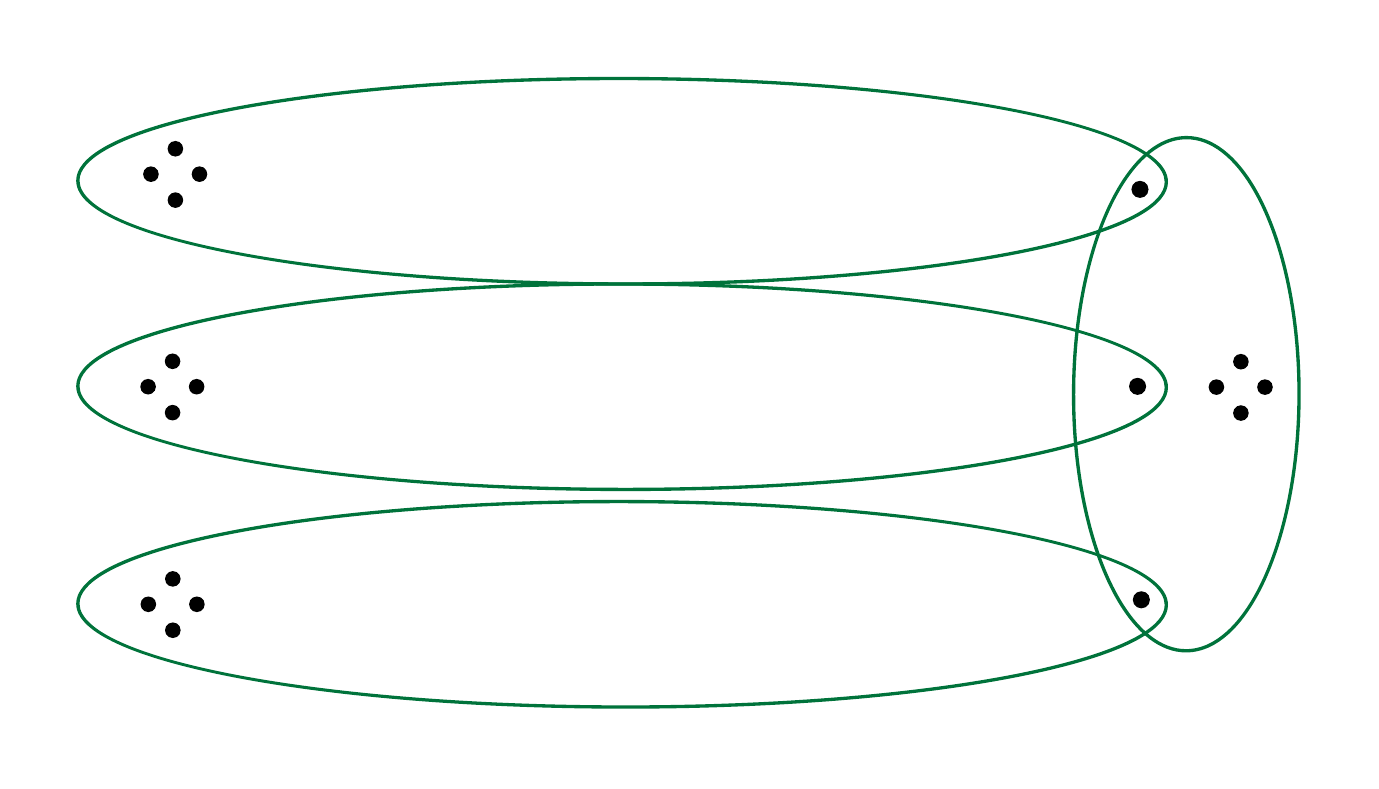} 
\caption{A bad instance for algorithms \gelms and \lsi 
for the intra-cluster dispersion problems (\icd and \icdq).}
\label{fig:intra_bad_case}
\end{figure}
We prove this argument by the example shown in Figure \ref{fig:intra_bad_case}, 
which shows 4 overlapping clusters. 
The horizontal distance of the three clusters on the left can be arbitrarily large. 
If \gelms considers the cluster on the right before any of the clusters of the left, 
the approximation ratio can be arbitrarily large.
For the \lsi algorithm, 
a selection containing any overlapping element can be a feasible initial selection for the right cluster, 
leading again to an arbitrarily large approximation ratio.
\hfill\end{proof}
From the proof of Lemma~\ref{lemma:unbounded}, it is worth noting
that the performance of the two approaches, \gelms and \lsi, 
depends on the order of processing the data clusters, 
and the initial solution, respectively.
Notice that the dependence of \gelms on the order of processing the clusters is unavoidable, 
as the first element of each cluster that gets added to the solution 
contributes nothing to the objective function.

\spara{Proposed method.}
We now introduce our proposed method for the \icd problem.
The main idea is to overcome the drawback of \gelms
and its dependence on cluster ordering. 
We can achieve this by finding a {\em pair of elements}, 
which both belong to the same cluster, 
and add them both to the solution. 
After adding a pair of elements to a set $\sel_j$, 
the corresponding budget $\nsel_j$ is reduced by 2. 
A desirable feature of our algorithm, 
is that at each step we can select the pair of elements
with the maximal contribution to some intuitive measure;
the order of processing clusters is now determined completely 
by the algorithm, and it is not arbitrary. 
\note[Guangyi]{the pair of elements with the maximal contribution to the objective is discussed in \cite{birnbaum2009improved}, where they extend single-element greedy to multiple-element greedy.}
Pseudocode of our algorithm, which we call \gp,
is shown in Algorithm~\ref{alg:intra}.

\noindent
\rule{\columnwidth}{0.001mm}

\vspace{-3mm}
\begin{alg}{\label{alg:intra}\gp}\\
\textbf{Input:}
Ground set $\univset$, 
clustering $\clustering=\{\cls_1,\ldots,\cls_{\ncls}\}$,
upper bounds $\nsel_1,\ldots,\nsel_{\ncls}$.
\\
\textbf{Output:} 
Intra-cluster dispersion set $\sel=\{\sel_1,\ldots,\sel_{\ncls}\}$.
\begin{algorithmic}[1]
\For {$j=1,\ldots,\ncls$} 
\State $\sel_j\gets\emptyset$;
$\pr{\nsel}_j \gets 2\lfloor\frac{\nsel_j}{2}\rfloor$
\EndFor
\While{(exists $j$ such that $|\sel_j|<\pr{\nsel}_j$)}
	\State $R\gets \{ j\in[\ncls] \text{ s.t.\ } |\sel_j|<\pr{\nsel}_j \}$
	\State Select a pair $\edge=\{u,v\}$ over clusters in $R$ whose both elements $u$ and $v$ belong in some cluster $\cls_{j}$ and $(\nsel_{j}-1)\,\dist(\edge)$ is maximized \label{alg:step:maxe}
	\State $\sel_{j} \gets \sel_{j} \cup \{u,v\}$
	\State $\univset \gets \univset \setminus \{u,v\}$
\EndWhile
\For {$j=1,\ldots,\ncls$} \If{($\nsel_j$ is odd)}
	\State Add an arbitrary element $v\in \cls_j$ into $\sel_j$ \label{alg:step:odd}
\EndIf \EndFor
\State {\bf Return} $\sel_1,\ldots,\sel_\ncls$
\end{algorithmic}
\end{alg}
\vspace{-5mm}
\noindent
\rule{\columnwidth}{0.001mm} 

\smallskip
Our first result is to demonstrate the relative power of 
greedily adding pairs of elements, instead of individual elements. 
In particular, in contrast to Lemma~\ref{lemma:unbounded} where we showed 
that \gelms has an unbounded approximation ratio, 
\gp offers a provable constant-factor approximation.
\begin{thm}
\label{thm:intra}
Algorithm \gp is 6-appro\-xi\-ma\-tion for \icd problem.
\end{thm}
\note[Guangyi]{The different b/w Thm \ref{thm:intra} and Thm \ref{thm:intra-subm} is that the former works even when $\nsel_j$ is odd, by leveraging the proof of Hassin et al.~\cite{hassin1997approximation} for a singleton cluster.
The mapping onto the selected pairs is slightly different.
Here, we only need to map edges in \opj that coincides with \sel.}
\begin{proof}
For simplicity, we assume $\dist(S_j)$ counts each pair once without loss of generality.
Let $\op$ be an optimal solution to an instance of the \icd problem, 
and let $\sel$ be the solution computed by \gp.
Let ${\opj}$ by the subset of the optimal solution in cluster $\cls_j$, 
so that $\op=\bigcup_{j=1}^{\ncls} {\opj}$,
and sim\-ilar\-ly $\sel=\bigcup_{j=1}^{\ncls} \sel_j$.
Let $\bar{{\opj}}$ be the intersection of ${\opj}$ with the elements of the solution $\sel$ other than $\sel_j$, 
i.e., $\bar{{\opj}}={\opj} \cap (\sel\setminus\sel_j)$. 

If $\bar{{\opj}}=\emptyset$, 
then \gp can be analyzed se\-pa\-ra\-te\-ly for each cluster $\cls_j$, 
as by Hassin et al.~\cite{hassin1997approximation}, 
and get $\dist(\sel_j) \ge \frac 12 \dist({\opj})$, 
obtaining a 2-approximation.
Same reasoning applies when $\bar{{\opj}}$ is defined to also exclude items in ${\opj} \cap (\sel\setminus\sel_j)$ that are picked up by other clusters (i.e., $\sel\setminus\sel_j$) \emph{after} $\sel_j$ is saturated (i.e., $2\lfloor\nsel_j/2\rfloor$ pairs selected).

If $\bar{{\opj}}\ne\emptyset$ and every item in it is picked up before $\sel_j$ is saturated, then 
$\dist(\sel_j) \ge \frac 12 \dist({\opj}\setminus\bar{{\opj}})$ still holds.
Additionally, we have to bound the remaining part $\dist({\opj})-\dist({\opj}\setminus\bar{{\opj}})$.
By the definition of distance \dist over set of elements, we have
$
\dist({\opj}) - \dist({\opj}\setminus\bar{{\opj}}) = - \dist(\bar{{\opj}}) + \sum_{v\in\bar{{\opj}}} \dist(v,{\opj}) \le \sum_{v\in\bar{{\opj}}} \dist(v,{\opj}).
$
Consider now an element $v\in\bar{{\opj}}$ and a pairwise distance $\dist(v,u)$ of it, 
Let $\{x,y\}$ be the first selected pair in $\sel$ of some cluster $\cls_i$ such that $\{x,y\}\cap\{u,v\}\ne\emptyset$. 
Then by greedy,
\eq{
\dist(v,u) \le (\nsel_i-1)/(\nsel_j-1) \, \dist(x,y).
}
For each pair $\{x,y\}$, when $x$ coincides with some object in \opj, $(\nsel_i-1)/(\nsel_j-1)\dist(x,y)$ may be needed at most $\nsel_j-1$ times,
and when $y$ also coincides, it is needed at most $2(\nsel_j-1)$ times in total, i.e., $2(\nsel_i-1)\dist(x,y)$ for each selected pair $\{x,y\}$ is sufficient.
By triangle inequality, we have
\eq{
(\nsel_i-1) \, \dist(x,y) \le  \dist(x,y) + \sum_{w\in(\sel_{i}\setminus\{x,y\})} \dist(x,w) + \dist(w,y),
}
so each $\dist(x,y)$ is only needed at most twice, but $\dist(x,w)$ between two selected pairs may be used up to four times.
Summing over all pairwise distances of all elements $v\in\bar{\op_j}$ and all clusters we~get
\eq{
\sum_{j=1}^{\ncls} \sum_{v\in\bar{{\opj}}} \dist(v,{\opj})  
\le 
4 \sum_{j=1}^{\ncls} \dist(\sel_j) = 4\, \distintra(\sel).
}
We can conclude that
\begin{align*}
6\,\distintra(\sel) & \ge 
\sum_{j=1}^{\ncls} \dist({\opj}\setminus\bar{{\opj}}) + \sum_{j=1}^{\ncls} \sum_{v\in\bar{{\opj}}} \dist(v,{\opj}) \\
&  \ge \distintra(\op).
\end{align*}
\hfill
\end{proof}

Furthermore, we can show that our analysis of \gp is tight.
\begin{rmk}
\label{rmk:intra}
The 6-approximation bound obtained in Theorem~\ref{thm:intra}
for \gp is tight, even if the distance function \dist is a metric.
\end{rmk}
\begin{proof}
We present a tight example in the Appendix.
\hfill\end{proof}

\section{Dispersion with submodular quality function}\label{sec:subm}
\note[Guangyi]{
For the modular case, the trick used in \cite{gollapudi2009axiomatic} does not work in multiple clusters.
In one cluster, we can simply replace $\dist(\cdot)$ with a new one $\pr{\dist}(\cdot)$, and we are good to go.
\begin{align*}
\pr{\dist}(x,y) = \frac{1}{\lambdasubm(\nsel-1)} \quality(\{x,y\}) + \dist(x,y)
\end{align*}
However, with multiple clusters, we can't do this anymore as we may have different $\nsel_j$...
}

We now turn our attention to the \icdq problem, 
with the quality function \quality being monotone non-decreasing and submodular. 
Notice that non-negative modular functions serve as a special case of such submodular functions, so we will devote our discussion only to the submodular case.
In this case, the \gp algorithm gives an equally strong guarantee 
with the following intuitive modification: 
In the $i$-{th} iteration we select the $i$-{th} 
pair of elements $\{u,v\}$ in some cluster $j$, 
among all unsaturated clusters, so as to maximize
\begin{equation}
\label{eq:alg:intra:maxedge:subm}
\feva_{\{u,v\}}(\sel^{i-1}) = 
	\quality_{\{u,v\}}(\sel^{i-1}) + \lambdasubm\,2(\pr{\nsel}_{j}-1)\dist(u,v),
\end{equation}
\note[Guangyi]{Mind the difference, $2\dist(u,v)=\dist(\{u,v\})$.}
where $\pr{\nsel}_{j}= 2\lceil\nsel_{j}/2\rceil$ and
$\quality_{\{u,v\}}(\setofedge^{i-1}) = \quality(\setofedge^{i-1}\cup\{u,v\}) - \quality(\setofedge^{i-1})$
is the marginal gain to the quality function \quality from adding the pair $\{u,v\}$ 
into the current solution $\setofedge^{i-1}$.
In our analysis we will make the mild assumption that the budgets $\nsel_j$ are even.
We propose a simple heuristic to deal with the case that some budget $\nsel_j$ is odd, 
but the approximation guarantee becomes slightly weaker.

Before we introduce our main theorem, we show the performance of the algorithm for a single cluster. 
The proof is deferred to the Appendix.
\begin{lem}
\label{lemma:subm-one-cluster}
The \gp algorithm with greedy rule~(\ref{eq:alg:intra:maxedge:subm})
gives a 4-approximation 
(or $(4 \min \{ \frac{\nsel+1}{\nsel-1}, 2 \})$-approximation if $\nsel$ is odd)
for the \scdq problem. 
\end{lem}

Notice that in the case of a single cluster, the \gelms algorithm is superior 
giving a 2-ap\-prox\-i\-ma\-tion~\cite{borodin2012max}.
However, when many clusters are present, the ap\-prox\-i\-ma\-tion ratio of \gelms becomes unbounded
(Lemma~\ref{lemma:unbounded}), 
while \gp maintains a constant-factor ap\-prox\-i\-ma\-tion guarantee, as we show next.

\begin{thm}
\label{thm:intra-subm}
The \gp algorithm with greedy rule~(\ref{eq:alg:intra:maxedge:subm})
gives a 6-approximation for the \icdq problem, 
when all $\nsel_j$ are even.
\end{thm}
\begin{proof}
We denote the optimal solution as $\op=\cup_{j=1}^{\ncls} {\opj}$, 
and the solution found by \gp as $\setofedge = \cup_{j=1}^{\ncls} \setofedge_j$.
Also, let $\{\edge_i\}$ be the set of pais of elements selected by \gp.
As in the proof of Theorem~\ref{thm:intra} we define
$\bar{\opj} = {\opj} \cap (\setofedge\setminus\setofedge_j)$.
If $\bar{\opj} = \emptyset$ for every cluster $\cls_j$, 
then by Lemma \ref{lemma:subm-one-cluster}, 
each set $\setofedge_j$ is 4-approximation with respect to ${\opj}$.
Hence $\setofedge$ is 4-approximation.

Thus, we proceed to analyze the case that $\bar{\opj} \ne\emptyset$ for some $j$.
For each set ${\opj}$ we consider a {\em maximum matching} $\{\ops{\edge}_{j\pr{i}}\}$.%
\footnote{Those maximum matchings are only to be considered conceptually, 
	for the purposes of the proof, not computed, as \op is not known.}
We define a mapping $\pi$ from $\{\ops{\edge}_{j\pr{i}}\}$ to $\{\edge_i\}$, as follows:
\squishlist
\item if $\ops{\edge}_{j\pr{i}}\cap\edge_i\ne\emptyset$ for some $i$ and cluster $j$ is not saturated before the $i$-th iteration, 
then $\pi(\ops{\edge}_{j\pr{i}})=\edge_i$ where the smallest $i$ is chosen;
\item otherwise $\ops{\edge}_{j\pr{i}}$ is mapped to some infrequently mapped pair $\edge_i$ of $\setofedge_j$.
\squishend
The key point of mapping $\pi$ is that each optimal pair $\ops{\edge}_{j\pr{i}}$ that falls into the second requirement
has to be mapped to some selected pair $\edge_i$ within cluster $j$, because pairs in other clusters can not be guaranteed to be better than $\ops{\edge}_{j\pr{i}}$, as they may be selected simply because cluster $j$ is saturated, i.e., a selected pair is greedy only among unsaturated clusters in its iteration.
Notice that due to such a requirement, 
each pair $\edge_i$ can be the image of at most three pairs $\ops{\edge}_{j\pr{i}}$:
two from the two pairs that intersect with $\edge_i$ and one from some pair that does not intersect with 
$\edge_i$ but is in the same cluster as $\edge_i$.
As each pair $\edge_i$ achieves the best marginal gain at the time it is selected, we obtain the following inequality:
\eq{
\phi_{\edge_i}(\setofedge^{i-1}) \ge \phi_{\ops{\edge}_{j\pr{i}}}(\setofedge^{i-1}) \ge \phi_{\ops{\edge}_{j\pr{i}}}(\setofedge).
}
The second inequality above is due to submodularity.

Assume that a total of $\nselsum$ elements are selected from all clusters. We have
\eq{
3 \sum_{i=1}^{\nselsum/2} \eqand \phi_{\edge_i}(\setofedge^{i-1})
\ge \sum_{j=1}^{\ncls}\sum_{\pr{i}=1}^{\nsel_j/2} \phi_{\ops{\edge}_{j\pr{i}}}(\setofedge) \eqnl
\eqand = \sum_{j=1}^{\ncls}\sum_{\pr{i}=1}^{\nsel_j/2} \quality(\setofedge\cup\ops{\edge}_{j\pr{i}}) - \quality(\setofedge) + \lambda 2(\nsel_j-1)\dist(\ops{\edge}_{j\pr{i}}) \eqnl
\eqand \ge \quality(\setofedge\cup\op) - \quality(\setofedge) + \sum_{j=1}^{\ncls}\sum_{\pr{i}=1}^{\nsel_j/2} \lambda 2(\nsel_j-1)\dist(\ops{\edge}_{j\pr{i}}) \eqnl
\eqand \ge \quality(\op) - \quality(\setofedge) + \sum_{j=1}^{\ncls}\sum_{\pr{i}=1}^{\nsel_j/2} \lambda 2(\nsel_j-1)\dist(\ops{\edge}_{j\pr{i}}),
}
where the last inequality uses the monotonicity of \quality.
\note[Guangyi]{monotonicity is used in the last step above.}
Putting together that
\eq{
3 \sum_{i=1}^{\nselsum/2} \phi_{\edge_i}(\setofedge^{i-1}) = 3 \quality(\setofedge) + \lambda\sum_{i=1}^{\nselsum/2} 6(\nsel_{\fedgetoclsnum(\edge_i)}-1) \dist(\edge_{i}) 
}
where $\fedgetoclsnum$ maps a pair $\edge_i$ to its cluster number, 
the inequalities 
$2(\nsel_j-1)\sum_{\pr{i}=1}^{\nsel_j/2} \dist(\ops{\edge}_{j\pr{i}})\ge \dist({\opj})$ 
(by the virtue of $\{\ops{\edge}_{j\pr{i}}\}$ being a maximum matching)
and 
$\dist(\setofedge_j) \ge \nsel_j \sum_{i=1}^{\nsel_j/2} \dist(\tilde{\edge}_{ji})$ 
(by triangular inequality),
where $\{\tilde{\edge}_{ji}\}$ can be an arbitrary set of disjoint pairs in $\setofedge_j$, 
we eventually have
\eq{
6  \distintrasubm(\setofedge) 
\eqand = 6 \quality(\setofedge) + \lambda 6\sum_{j=1}^k \dist(\setofedge_j) \eqnl
\eqand \ge 6\quality(\setofedge) + \lambda \sum_{i=1}^{\nselsum/2} 6\nsel_{\fedgetoclsnum(\edge_i)} \dist(\edge_{i}) \eqnl
\eqand \ge 4\quality(\setofedge) + \lambda \sum_{i=1}^{\nselsum/2} 6(\nsel_{\fedgetoclsnum(\edge_i)}-1) \dist(\edge_{i}) \eqnl
\eqand \ge \quality(\op) + \sum_{j=1}^{\ncls}\sum_{\pr{i}=1}^{\nsel_j/2} \lambda 2(\nsel_j-1)\dist(\ops{\edge}_{j\pr{i}}) \ge \distintrasubm(\op).
}
As modular functions is a subset of submodular functions, 
Remark~\ref{rmk:intra} applies, 
and thus, our bound for \gp is tight.%
\hfill\end{proof}

When some $\nsel_j$ is odd, 
we run the algorithm to select $2 \lceil \nsel_j/2 \rceil$ elements for each cluster, and then we remove an element in each cluster with an odd $\nsel_j$ that contributes the least with respect to some measure.
A guarantee of $(6\min\{ \frac{\nselmin+1}{\nselmin-1},2 \})$-approximation can be proved,
where $\nselmin$ is the smallest odd $\nsel_j$.
We defer the details to the Appendix.

\section{Scalable algorithms}\label{sec:scalability}
As the data size increases, a vanilla implementation of our methods may be impractical. 
In this section we ex\-am\-ine the scalability of our algorithm, and discuss ex\-ten\-sions to improve its efficiency.
We focus on scaling~up the algorithm on a single machine; 
we are not after parallel or distributed implementations, which are of independent interest,  
but beyond the scope of this paper. 
We first inspect the time complexity of the three methods we have presented.
\begin{center}
\begin{small}
\begin{tabular}{ll}
\toprule
Method & Time complexity \\
\midrule
\gelms & $\bigO(\ncls \nselmax^2 \nunivset)$\\
\lsi   & $\bigO(\frac{1}{\lseps} \ncls \nselmax^2 \nunivset )$\\
\gp    & $\bigO(\ncls \nselmax \nunivset^2)$ \\
\bottomrule
\end{tabular}
\end{small}
\end{center}
Recall that  
\nunivset is the number of elements in~\univset, 
\ncls is the number of clusters in~\clustering, 
and 
\nselmax is the maximum budget over all $\nsel_j$. 
The accuracy parameter \lseps controls the termination of the \lsi algorithm.

In a practical scenario we expect \nunivset to be much larger than \ncls and \nselmax, 
thus, we are interested in optimizing time complexity with respect to~\nunivset.
From the above table we can see that \gelms and \lsi
\footnote{The optimal objective value is assumed to be a constant, and a swap is performed only when it increases the objective by more than $\epsilon$} 
are both linear in~\nunivset, 
though the latter is slower in practice due to the need for convergence.
The quadratic complexity of \gp is due to the need of finding the {\em furthest pair} of elements, 
which we also refer to as {\em diameter}.
To cope with the quadratic complexity, 
we borrow a well-known $2$-approximation algorithm for finding the diameter in a metric space.
Hereinafter, we use the terms ``pair'' and ``edge'' interchangeably.
\begin{lem}
\label{lemma:diameter}
Picking an element $x$ arbitrarily, and the element $y$ that is the furthest from $x$, 
gives a $2$-approx\-imation to the problem of finding the diameter of a set of elements in a metric space.
\end{lem}
\begin{proof}
If optimal diameter is $(u,v)$, then 
$\dist(u,v) \le \dist(u,x) + \dist(x,v) \le 2\dist(x,y)$.
\hfill\end{proof}

The idea behind Lemma~\ref{lemma:diameter} can be used not only to improve the running time
but also to provide additional flexibility in adding elements to our solution, 
and obtaining solutions of higher quality in practice.

In particular, for the case of maximizing dispersion in a single cluster (problem \scd), 
if ${\sel}^{i-1}$ is the set of currently-selected elements, 
in the next iteration
we can select the first element $x$ to be the one that maximizes $\dist(x,{\sel}^{i-1})$.
For the second element $y$, 
instead of selecting the one that is the furthest from $x$ (as Lemma~\ref{lemma:diameter} dictates), 
we have observed that, in practice, it is better to  
find an element $y$ that balances its distance between $x$ and ${\sel}^{i-1}$, 
via a hyper-parameter $\param\in[0,1]$.
In particular, if $y^*$ is the furthest element from $x$, 
we find $y$ so that $\dist(x,y)\ge\param\dist(x,y^*)$
and  $\dist(y,{\sel}^{i-1})$ is max\-imized.
We refer to this version of \gp as \gpa.

\gps{\param=1} is the approximate version of \gp,
which selects the diameter using the 2-approximation method of Lemma~\ref{lemma:diameter}.
\gps{\param=0} de\-gen\-er\-ates into \gelms.
For intermediate values of $\param$, \gpa combines the best of both worlds:
the theoretical properties of \gp and the em\-pir\-i\-cal effectiveness of \gelms.

\begin{lem}\label{lem:gpa:1cls}
The \gpa algorithm with $0<\param\le 1$,
is a $(4/\param)$-ap\-proximation for the \scd problem.
\end{lem}
\begin{proof}
This argument can be proved by induction, 
as in the work of Hassin et al.~\cite{hassin1997approximation}.
Details are given in the Appendix.
\hfill\end{proof}

The case of maximizing dispersion for multiple clusters (problem \icd) is similar. 
In each iteration, we select a furthest pair in each unsaturated cluster,
and add to the solution the best pair out of these candidates. 

\begin{thm}\label{thm:gpa:multi}
The \gpa algorithm with $0<\param\le 1$,
is a $(12/\param)$-approximation for the \icd problem.
\end{thm}
\begin{proof}
The proof is similar to that in Theorem \ref{thm:intra}.
Details are given in the Appendix.
\hfill\end{proof}

The situation is more complicated when a quality function \quality is involved.
As before we focus on monotone non-decreasing and submodular functions.
First notice that evaluating Equation~(\ref{eq:alg:intra:maxedge:subm}) 
has quadratic complexity, as we need to optimize quality as well, 
so we need another fast heuristic. 
We use $\hat{\edge}$ to represent a 2-approximation diameter, and $\edge$ for an exact one.
In this case, we select the first endpoint $x$ that maximizes $\quality_{x}(\hat{\setofedge}^{i-1})$ in the $i$-th iteration, and then search for the best pair $\hat{\edge}_i^x$ with respect to 
$\feva_{\hat{\edge}_i^x}(\hat{\setofedge}^{i-1})$. 
Note that this pair is not necessarily the maximum-distance pair $\hat{\edge}^{\text{max},x}_{i}$ that includes $x$.
As a result:
\eq{
2  \feva_{\hat{\edge}_i^x}(\hat{\setofedge}^{i-1})  \ge
2  \feva_{\hat{\edge}^{\text{max},x}_i}(\hat{\setofedge}^{i-1})  \ge
 \feva_{\edge_i}(\hat{\setofedge}^{i-1}) . 
}
Similarly, we can still apply the same trick to look for a first endpoint $\pr{x}$ of quality at least $\param \quality_{x}(\hat{\setofedge}^{i-1})$ and a second endpoint $\pr{y}$ that gains an additional value at least 
$\param (\feva_{\hat{\edge}_i^{\pr{x}}}(\hat{\setofedge}^{i-1}) - \quality_{\pr{x}}(\hat{\setofedge}^{i-1}))$.
Hence,
\eq{
\frac 2\param   \feva_{(\pr{x},\pr{y})}(\hat{\setofedge}^{i-1}) \ge
\frac 2\param   \feva_{\hat{\edge}^{\text{max},\pr{x}}_i}(\hat{\setofedge}^{i-1})  \ge \feva_{\edge_i}(\hat{\setofedge}^{i-1}).
}
We can show the following.
\begin{lem}\label{lem:gpa:1cls:subm}
The \gpa algorithm with $0<\param\le 1$ and a greedy rule in Equation~(\ref{eq:alg:intra:maxedge:subm}) is $(6/\param)$-approximation for the \scdq problem when $\nsel$ is even.
\end{lem}
\begin{proof}
The proof is similar to Theorem \ref{thm:intra-subm}, the details can be found in the Appendix.
\hfill\end{proof}

With multiple clusters, the selection scheme is the same as its modular counterpart, but with a different objective, that is, we select a best pair as in Lemma \ref{lem:gpa:1cls:subm} within each unsaturated cluster and pick the final best one among them.
\begin{thm}\label{thm:gpa:multi:subm}
The \gpa algorithm with $0<\param\le 1$ and a greedy rule in Equation~(\ref{eq:alg:intra:maxedge:subm}) is $(12/\param)$-approximation  for the \icdq problem when $\nsel_j$ is all even.
\end{thm}
\begin{proof}
The proof is similar to that for one single cluster. 
Details are given in the Appendix.
\hfill\end{proof}

When some $\nsel_j$ is odd, we can apply the same technique as in Section \ref{sec:subm}, and the approximation ratio  increases only by a multiplicative factor of $\min\{ \frac{\nselmin+1}{\nselmin-1},2 \}$, where $\nselmin$ is the smallest odd $\nsel_j$.

\section{Experiments}\label{sec:experiment}
We evaluate our methods using six datasets, two synthetic, and four real-world data,
to examine their ability to obtain a selection of both good quality and diversity.
We use two of the real-world datasets to evaluate dispersion alone, and 
the other two to evaluate dispersion combined with a submodular quality function. 
All baselines are used in the comparison with small datasets, and only 
the scalable ones are applied to large datasets.

The list of all the algorithms we evaluated, and their abbreviations, is the following:
\gpa (\agpa) --- the proposed algorithm;
\gelms (\agelms),
\lsg (\alsg),
\lsi (\alsi),
\mc (\amc), and
\rn (\arn).
Unlike our algorithm \gpa, all baselines optimize the objective locally, 
so we will run them multiple times with input clusters in different order.
We run our algorithm with several different \param as well.
min/avg/max of multiple runs are reported.
We also investigate the sensitivity of \gpa to its hyperparameter $\param$.

The implementation of our algorithm and all baselines
is available online.\footnote{\href{https://github.com/Guangyi-Zhang/clustered-max-diversity}{https://github.com/Guangyi-Zhang/clustered-max-diversity}}

\spara{Results on synthetic datasets.}
The first dataset is a set of random vectors, each assigned to a fixed number of clusters randomly.
The second one tries to mimic structure of coherent clusters in real data by creating Normally-distributed random vectors around predefined center vectors, one for each cluster (aka prototypes). 
In order to generate overlapping clusters, we also assign a vector to its closest prototype.

We first investigate sensitivity of \gpa to its hyperparameter.
In both of the synthetic datasets, random and prototype, 
a larger $\param$, i.e., trying to find a diameter almost as good as the approximated one, 
is likely to have better performance than smaller ones.
Since our algorithm \gpa draws closer to \gelms with a smaller $\param$,  
this phenomenon also reflects the effectiveness of \gpa.
For more detailed results we refer to the Supplementary material, Table~\sref{tbl:param}.

In the small synthetic datasets shown in Table \ref{tbl:syn-small}, \gpa outperforms all baselines in the random dataset, and unsurprisingly, \rn has the worst performance. \lsg also performs poorly as it is optimizing a different objective, global dispersion.
More interesting results are observed in the prototype dataset. When the number of selected elements is small, \gpa surpasses all others. However, when the number of selected points becomes extremely large, a random initialization for \lsi and \lsg obtains the best result as it does not favor any elements, while other algorithms have a preference for overlapping elements, causing some clusters to not be able to select enough elements.
Interestingly, \rn performs equally well in this case, indicating 
the lack of structure in these problem instances.

It is worth noting that \gpa is quite robust to such extreme cases as it select best pairs globally.
The performance remains similar with data of a larger scale.
Detailed results are deferred to the supplementary material, Table~\sref{tbl:syn}. 
Similarly, validation on the linear-time scalability of \gpa is shown in the supplementary material, 
Figure~\sref{fig:time}.

\begin{table*}
\footnotesize
\caption{Performance for small synthetic datasets ($n=1000$, 10 clusters) where $dim$ stands for dimension.}
\label{tbl:syn-small}
\begin{tabular}{HlHrrrrrrrrrrrrr}
\toprule
{} & Dataset &     n &    b &   dim &  GPmin &  GPavg &  GPmax &  GVmin &  GVavg &  GVmax &  LSImin &  LSIavg &  LSImax &    LSG &     RN \\
exp\_id &         &       &      &     &        &        &        &        &        &        &         &         &         &        &        \\
\midrule
503    &  random &  1000 &   10 &   2 &  0.987 &  0.993 &    \bf{1.0} &  0.975 &  0.984 &  0.988 &   0.970 &   0.975 &   0.979 &  0.942 &  0.656 \\
504    &  random &  1000 &   10 &  10 &  0.993 &  0.996 &    \bf{1.0} &  0.983 &  0.987 &  0.991 &   0.987 &   0.993 &   0.996 &  0.962 &  0.792 \\
505    &  random &  1000 &  100 &   2 &  0.975 &  0.987 &    \bf{1.0} &  0.923 &  0.931 &  0.935 &   0.944 &   0.955 &   0.972 &  0.949 &  0.918 \\
506    &  random &  1000 &  100 &  10 &  0.957 &  0.981 &    \bf{1.0} &  0.935 &  0.941 &  0.944 &   0.925 &   0.935 &   0.942 &  0.937 &  0.919 \\
\midrule
471    &   proto &  1000 &   10 &   2 &  0.975 &  0.989 &  \bf{1.000} &  0.963 &  0.978 &  0.984 &   0.950 &   0.956 &   0.964 &  0.579 &  0.405 \\
473    &   proto &  1000 &   10 &  10 &  0.993 &  0.998 &  \bf{1.000} &  0.978 &  0.988 &  0.993 &   0.992 &   0.994 &   0.997 &  0.804 &  0.751 \\
475    &   proto &  1000 &  100 &   2 &  0.957 &  0.972 &  0.979 &  0.798 &  0.810 &  0.826 &   0.876 &   0.934 &   \bf{1.000} &  \bf{1.000} &  \bf{1.000} \\
476    &   proto &  1000 &  100 &  10 &  0.987 &  0.992 &  0.997 &  0.985 &  0.986 &  0.987 &   0.994 &   0.998 &   \bf{1.000} &  \bf{1.000} &  \bf{1.000} \\
\bottomrule
\end{tabular}
\end{table*}

\spara{Results on real-world data: Topical document aggregation.}
In order to create overlapping topics, we employ LDA \cite{hoffman2010online} on two document datasets, 2\,797 documents from five topics in \textit{20NewsGroups} dataset and 18\,713 blogs, one for each blogger \cite{schler2006effects}, and assign each document to the topics to which estimated probability is larger than a threshold.
Note that Cosine distance is computed between Word2Vec vectors of documents, 
which is a metric for normalized vectors.
We observe that \gpa surpasses the strong baseline \gelms in various data settings, 
showing a superior ability to select diverse documents.
Due to space limitation, the results are shown in the supplementary material, Table~\sref{tbl:doc}.

\spara{Results on real-world data: Descriptors for scho\-lar communities and movie categories.}
Upon the scholar collaboration network provided by \textit{Aminer} \cite{tang2008arnetminer}, we run K-clique community detection algorithm \cite{palla2005uncovering} to discover dense and overlapping scholar communities.
In this dataset, each scholar is associated with a set of keywords, representing their main research interests.
To be more specific, there are 544 communities formed by 6\,529 scholars and 23\,827 keywords.
Our goal is to select a small number of keywords for each community in the hope of fully characterizing it, in a way that the keywords cover as many scholars in its community as possible and they each represent distinct concepts.
In more detail, 
the quality of a set of keywords is defined by the cardinality of their coverage over all scholars, and Jaccard distance is applied between different keywords.
The reason we quantify the quality by global coverage instead of intra-cluster coverage is that the quality function is not submodular with the latter formulation. A simple technique to avoid selecting a keyword that has large global coverage from a cluster on which it has small coverage is by preprocessing, i.e., removing a keyword from clusters it covers weakly.

Another dataset we adopt is \textit{Movielens} \cite{harper2016movielens}, consisting 27\,278 movies, each falling into several of 20 categories and described by many short comments from different users.
We extract adjectives from comments to represent user emotions or opinions.
Our goal is to describe each category by diverse user opinions.
We formulate quality of an adjective as the number of movies it describes, and distance between two adjectives as their Cosine distance in Word2Vec embedding space.

In order to evaluate our methods both for quality and dispersion, 
we compare \gpa with two baselines, \gelms and \mc, the latter of which is specialized for quality.
We vary two model parameters, the number of selected elements ($\nsel$) and the bi-objectives trade-off ($\lambdasubm$), with the increase of which, dispersion plays a more important role in the combined objective.
Besides, due to the large number of clusters, we will not run \gelms multiple times with input clusters of different order. Accordingly, we also fix the hyperparameter of \gpa at $\param=0.95$.

The result is shown in Table \ref{tbl:subm}, where \gpa consistently achieves the best performance with respect to a combined objective under different model parameters. 
Though it can be expected that \gpa will achieve a good compromise between these two objectives, the result turns out to be even better, as \gpa obtains the best quality and the best dispersion simultaneously until model parameters distort the balance between the two objectives too much.
The reason for this is that, \mc fails to select enough elements in highly overlapping clusters, while \gpa manages to do so, in other words, maintaining a balance among different clusters.

A use case for scholar communities is shown in the supplementary material in Table~\sref{tbl:usecase}.

\begin{table*}[h]
\footnotesize
\centering
\caption{Performance on clustering descriptors selection, where columns starting with $Q$ means quality, one with $D$ means dispersion and one with $QD$ means a combination of both.}
\begin{tabular}{Hlrrrrr|rrr|rrrHHH}
\toprule
{} &  Dataset &   $\nsel$ &  $\lambdasubm$ &       QDGP &       QDGV &       QDMC &       DGP &       DGV &       DMC &     QGP &     QGV &     QMC &   PGP &   PGV &  PMC \\
exp\_id &          &     &        &            &            &            &           &           &           &         &         &         &       &       &      \\
\midrule
457    &  scholar &   6 &      1 &   \textbf{22258.25} &   22053.14 &   14949.15 &  \textbf{15734.25} &  15631.14 &   8452.15 &  \textbf{6524.0} &  6422.0 &  6497.0 &  6.00 &  6.00 &  5.0 \\
465    &  scholar &   6 &      5 &   \textbf{85159.28} &   85010.09 &   48757.75 &  \textbf{15728.06} &  15717.02 &   8452.15 &  \textbf{6519.0} &  6425.0 &  6497.0 &  6.00 &  6.00 &  5.0 \\
466    &  scholar &   6 &     10 &  163833.55 &  \textbf{163942.39} &   91018.50 &  15731.06 &  \textbf{15751.24} &   8452.15 &  \textbf{6523.0} &  6430.0 &  6497.0 &  6.00 &  6.00 &  5.0 \\
467    &  scholar &  10 &      1 &   \textbf{52925.29} &   52548.52 &   32461.70 &  \textbf{46396.29} &  46038.52 &  25939.70 &  \textbf{6529.0} &  6510.0 &  6522.0 &  9.98 &  9.99 &  9.0 \\
468    &  scholar &  10 &      5 &  \textbf{238424.28} &  237782.79 &  136220.49 &  \textbf{46379.06} &  46254.56 &  25939.70 &  \textbf{6529.0} &  6510.0 &  6522.0 &  9.98 &  9.99 &  9.0 \\
469    &  scholar &  10 &     10 &  \textbf{470291.14} &  469702.83 &  265918.98 &  \textbf{46376.21} &  46319.38 &  25939.70 &  \textbf{6529.0} &  6509.0 &  6522.0 &  9.98 &  9.99 &  9.0 \\
\midrule
448    &  movielens &  10 &      1 &  \textbf{10145.64} &  10017.05 &   9736.03 &  \textbf{1546.64} &  1490.05 &  1201.03 &  \textbf{8599.0} &  8527.0 &  8535.0 &   9.5 &   9.50 &   8.55 \\
451    &  movielens &  10 &      5 &  \textbf{16399.05} &  16073.75 &  14540.16 &  \textbf{1572.41} &  1510.15 &  1201.03 &  \textbf{8537.0} &  8523.0 &  8535.0 &   9.5 &   9.50 &   8.55 \\
453    &  movielens &  10 &     10 &  \textbf{24340.6}0 &  23799.07 &  20545.32 &  \textbf{1584.56} &  1527.31 &  1201.03 &  8495.0 &  8526.0 &  \textbf{8535.0} &   9.5 &   9.50 &   8.55 \\
454    &  movielens &  20 &      1 &  \textbf{15429.97} &  15079.33 &  14566.02 &  \textbf{6565.97} &  6254.33 &  5702.02 &  \textbf{8864.0} &  8825.0 &  \textbf{8864.0} &  19.0 &  18.85 &  18.05 \\
456    &  movielens &  20 &      5 &  \textbf{41906.11} &  40759.44 &  37374.10 &  \textbf{6614.82} &  6390.29 &  5702.02 &  8832.0 &  8808.0 &  \textbf{8864.0} &  19.0 &  18.90 &  18.05 \\
458    &  movielens &  20 &     10 &  \textbf{75123.33} &  72992.49 &  65884.19 &  \textbf{6629.73} &  6420.25 &  5702.02 &  8826.0 &  8790.0 &  \textbf{8864.0} &  19.0 &  18.85 &  18.05 \\
\bottomrule
\end{tabular}
\label{tbl:subm}
\end{table*}

\section{Conclusion}\label{sec:conclusion}
We propose a provable algorithm for intra-cluster dispersion and diversification problems,
for which popular heuristics fail to achieve any guarantee.
For the diversification version, the theoretical guarantee is achieved
when the quality function is modular and submodular, 
thus, accommodating a wide range of real-life applications.
Our algorithm selects pairs of elements greedily. 
To achieve a desirable linear-time complexity, 
the algorithm incorporates an approximate diameter computation, 
which also permits additional flexibility in choosing diverse pairs of endpoints.
In the empirical study on both synthetic and real-world datasets, 
our algorithm exhibits a superior ability to maximize dispersion and quality while maintaining balance among clusters.

Potential future directions include extensions to other quality functions and 
other notions of~dispersion.

\bibliographystyle{siam}
\bibliography{references}

\begin{appendices}

\section{Proofs}

\subsection{Proof for Remark \ref{rmk:intra}}

\begin{proof}
We present a tight example:
Consider an instance with $q+1$ clusters.
The first $q$ clusters have two hubs $v_{j1}$ and $v_{j2}$, $j=1,\ldots,q$, 
and other elements, which we denote by $\pr{\cls}_j$.
Distances are $d(v,u)=\epsilon$, for all $v,u\in \pr{\cls}_j$,
$d(v_{j1},v)=d(v_{j2},v)=2$, for all $v\in \pr{\cls}_j$,
and 
$d(v_{j1},v_{j2})=2+\epsilon$, 
where $\epsilon>0$ is infinitesimal.
The last cluster $\cls_{q+1}$ has $4q$ elements,
$2q$ of which form an equidistant clique $K_{2q}$ with distance~2, 
and the other $2q$ form a perfect matching $M_{2q}$, 
where each matching edge is formed by two hubs in one of previous $q$ clusters
(thus, the clusters overlap).
All other pairs of elements in $\cls_{q+1}$ have distance $1+\epsilon$.
Notice that $d$ is indeed a metric.

We consider the \icd problem with $\nsel_j=2q$, for all $j=1,\ldots,q+1$.
An optimal solution will include the 2 hubs for each of the first $q$ clusters, 
and the $K_{2q}$ clique in cluster $\cls_{q+1}$.
Instead, \gp will pick $M_{2q}$ for cluster $\cls_{q+1}$, and no-hub elements in the first $q$ clusters.
As $q$ increases,
the ratio between the two solutions becomes:
\begin{small}
	\begin{align*}
		\frac{\dist(\op)}{\dist(\sel)} 
		& = \frac{\dist(K_{2q})+\sum_{j=1}^{q} \dist({\opj})}{\dist(M_{2q})+\sum_{j=1}^{q} \dist(\sel_j)}\\
		& = \frac{2\,2q(2q-1)/2+ q\,((2+\epsilon)+2\,2(2q-2))}
		{q(2+\epsilon)+(2q(2q-2)/2)(1+\epsilon) + q\, 2q(2q-1)\epsilon/2}\\
		& = \frac{12q-8+\epsilon}{2q+(q+1)(2q-1)\epsilon} \approx 6.
	\end{align*}
\end{small}
\hfill\end{proof}

\subsection{Proof for Lemma \ref{lemma:subm-one-cluster}}

\begin{proof}
When $\nsel$ is even, the proof is similar to that in Theorem \ref{thm:intra-subm}, where we find a maximal matching $\{\ops{\edge}_{\pr{i}}\}$ in the optimal $\op$, and assign each pair of it to a selected pair $\edge_i$.
By a same mapping, we can obtain the following inequality.
\begin{align*}
\phi_{\edge_i}(\setofedge^{i-1}) \ge \phi_{\ops{\edge}_{\pr{i}}}(\setofedge^{i-1}) \ge \phi_{\ops{\edge}_{\pr{i}}}(\setofedge)
\end{align*}
Since each $\edge_i$ is mapped at most twice.
\begin{align*}
2\cdot \sum_{i=1}^{\nsel/2} \phi_{\edge_i}(\setofedge^{i-1}) &\ge \sum_{\pr{i}=1}^{\nsel/2} \phi_{\ops{\edge}_{\pr{i}}}(\setofedge)\\
3\quality(\setofedge) + 4\lambda(\nsel-1)\sum_{i=1}^{\nsel/2} \dist(\edge_{i}) &\ge \quality(\op) + 2\lambda(\nsel-1)\sum_{\pr{i}=1}^{\nsel/2} \dist(\ops{\edge}_{\pr{i}})\\
4\distsubm(\setofedge) &\ge \distsubm(\op)
\end{align*}

Now we go back to handle the case when $\nsel$ is odd.
First of all, in a trivial case when $\nsel=1$, we simply pick one element that has maximum quality and the solution is optimal.
Otherwise, we run the algorithm to select $\nsel+1$ elements, and then we remove the element that contributes the least with respect to a new measure.
Let elements in $\setofedge^{b+1}$ be ordered arbitrarily as $v_1,...,v_{b+1}$, and $\setofedge^{i}$ refers to as the set of first $i$ elements.
A new measure for each element is given as 
\begin{align*}
\fmgnv(v_i) = \quality_{v_i}(\setofedge^{i}\setminus v_i) + \lambda \dist(v_i, \setofedge^{b+1}\setminus v_i)
\end{align*}
It is easy to see that
\begin{align*}
\distsubm(\setofedge^{b+1}) = \sum_{i=1}^{b+1} \fmgnv(v_i)
\end{align*}
Say that the $l$-th element has the least value with respect to $\fmgnv(\cdot)$, which means $\fmgnv(v_{l}) \le \frac{1}{b+1} \distsubm(\setofedge^{b+1})$.
We denote $\setofedge = \setofedge^{b+1} \setminus v_{l}$, and we have,
\begin{align*}
&\frac{b-1}{b+1} \distsubm(\setofedge^{b+1}) \le
- \lambda \dist(v_{l}, \setofedge^{b+1}\setminus v_{l}) +  \sum_{i\in[b+1],i\ne l} \fmgnv(v_i) \\
&= \sum_{i\in[b+1],i\ne l} \quality_{v_i}(\setofedge^{i}\setminus v_i) + \lambda \dist(v_i, \setofedge \setminus v_i) \\
&\le \quality(\setofedge) +  \sum_{i\in[b+1],i\ne l} \lambda \dist(v_i, \setofedge \setminus v_i)
= \distsubm(\setofedge)
\end{align*}
Therefore
\begin{align*}
4\frac{b+1}{b-1} \distsubm(\setofedge) \ge 4\distsubm(\setofedge^{b+1}) \ge \distsubm(\op^{b+1}) \ge \distsubm(\op)
\end{align*}
\hfill\end{proof}

\subsection{Proof for Theorem \ref{thm:intra-subm} with odd $\nsel_j$}
\begin{proof}
Similar to that in Lemma \ref{lemma:subm-one-cluster}, we run the algorithm to select $\pr{\nsel_j}$ elements where $\pr{\nsel_j}=2 \lceil \nsel_j/2 \rceil$, and then we remove an element in each cluster with an odd $\nsel_j$ that contributes the least with respect to a similar measure as in Lemma \ref{lemma:subm-one-cluster}.
Let $\pr{\nselsum}=\sum_{j\in[m]} \pr{\nsel_j}$, and elements in ${{\setofedge}}^{\pr{\nselsum}}$ be ordered arbitrarily as $v_1,...,v_{\pr{\nselsum}}$, and ${{\setofedge}}^{i}$ refers to as the set of first $i$ elements.
The measure for each element is given as 
\begin{align*}
\fmgnv(v_i) = \quality_{v_i}({{\setofedge}}^{i}\setminus v_i) + \lambda \dist(v_i, {\pr{\setofedge}}_{\fedgetoclsnum(v_i)}\setminus v_i)
\end{align*}
Where $\fedgetoclsnum(\cdot)$ maps an selected element to its cluster number. It is easy to see that
\begin{align*}
\distintrasubm({{\setofedge}}^{\pr{\nselsum}}) = \sum_{i=1}^{\pr{\nselsum}} \fmgnv(v_i)
\end{align*}
Now we remove an element $v_{jl}$ in each cluster with an odd $\nsel_j$ that has the least value with respect to $\fmgnv(\cdot)$ within that cluster.
We denote $\setofedge_j = {\pr{\setofedge}}_j \setminus v_{jl}$ and $\setofedge = \sum_{j\in[m]} \setofedge_j$.
By similar derivation as in Lemma \ref{lemma:subm-one-cluster}, we obtain
\begin{align*}
\frac{{\nselmin}-1}{{\nselmin}+1} \distintrasubm({{\setofedge}}^{\pr{\nselsum}}) \le \distintrasubm(\setofedge)
\end{align*}
where ${\nselmin}$ is the smallest odd $\nsel_j$.
Therefore
\begin{align*}
6\frac{{\nselmin}+1}{{\nselmin}-1} \distsubm(\setofedge) \ge 6\distsubm(\setofedge^{\pr{\nselsum}}) \ge \distsubm(\op^{\pr{\nselsum}}) \ge \distsubm(\op)
\end{align*}
if there exists $\nselmin=1$, then the multiplicative factor becomes $\frac{{\nselmin}+1}{{\nselmin}}$ instead of $\frac{{\nselmin}+1}{{\nselmin}-1}$, because in the cluster $\nsel_j=1$, we do not need to care about dispersion as it is zero.
Basically the reason we have a form like $\frac{{\nsel_j}+1}{{\nsel_j}-1}$ is because when we remove one element from $\nsel_j+1$ elements, we need to reduce the objective value by $1/(\nsel_j+1)$ for quality and half dispersion removed, and by another $1/(\nsel_j+1)$ for the other half dispersion. Now the latter part is no longer needed.
\hfill\end{proof}

\subsection{Proof for Lemma \ref{lem:gpa:1cls}}

\begin{proof}
This argument can be proved by induction, 
as in the work of Hassin et al.~\cite{hassin1997approximation}.
The statement is trivially true when $\nsel=1,2$ and also true when $\nsel=3$ by triangular inequality.
For a larger $\nsel$, we set aside the first selected pair $\edge_1$,  and the approximation factor of $\sel\setminus\edge_1$ still holds for $\nsel-2$ in the universal set  $\univset\setminus\edge_1$.
Adding back $\edge_1$, dispersion increases by at least $(\nsel-1)\dist(\pr{\edge})$ where $\pr{\edge}$ is the diameter, while the optimal solution with two arbitrary elements removed increases by at most $(4\nsel-6)\dist(\pr{\edge})$ by adding them back.
\hfill\end{proof}

\subsection{Proof for Theorem \ref{thm:gpa:multi}.} 

Instead of directly proving Theorem \ref{thm:gpa:multi}, we prove an enhanced 10-approximation version, while the proof for Theorem \ref{thm:gpa:multi} is nearly identical.

In the enhanced version, in each iteration for adding a new pair, 
We select a first element $x$ from unsaturated clusters, and then choose a second element $y$ that is furthest from $x$ among all unsaturated clusters $x$ belongs to.
The reason for this is that we can ensure the pair selected is optimal for $x$ among all unsaturated clusters.
Since $x$ may not cover all unsaturated clusters, we will keep selecting another first element from uncovered unsaturated clusters and repeat the same procedure, until every unsaturated cluster is covered.
At the end we choose the best pair among these candidates, each associated with a first element.

\begin{proof}
The proof is similar to that in Theorem \ref{thm:intra}.
If $\bar{\opj}=\emptyset$ for every cluster $j$, then $\dist({\opj})\le 4 \dist(\hat{\setofedge}_j)$, leading to $\dist(\op)\le 4 \dist(\hat{\setofedge})$.

Otherwise, $\dist(\sel_j) \ge \frac 14 \dist({\opj}\setminus\bar{\opj})$ still holds.
We bound the remaining part in the same way.
\[
\dist({\opj}) - \dist({\opj}\setminus\bar{\opj}) \le 2\sum_{v\in\bar{\opj}} \dist(v,{\opj}).
\]
Consider now an element $v\in\bar{\opj}$
and let $\{u,v\}$ and $\{v,w\}$ be two possible selected pairs in $\sel$, the latter selecting $v$ as its first endpoint. Then
\begin{eqnarray*}
\dist(v,{\opj}) & \le & (\nsel_j-1) \, 2\dist(\{u,v\}) \\
\dist(v,{\opj}) & \le & (\nsel_j-1) \, \dist(\{v,w\})
\end{eqnarray*}
That is to say, while each selected pair may be mapped to by at most two elements in $\cup_{j=1}^{\ncls} \bar{\opj}$, it only needs three instead of four copies to handle them.
Summing over all $\bar{\opj}$ we~get
\[
\sum_{j=1}^{\ncls} \sum_{v\in\bar{\opj}} \dist(v,{\opj})  \le
3 \sum_{j=1}^{\ncls} \dist(\sel_j) = 3\, \distintra(\sel),
\]
We can conclude,
\[
10\,\distintra(\sel) \ge
\sum_{j=1}^{\ncls} \dist({\opj}\setminus\bar{\opj}) + 2\sum_{j=1}^{\ncls} \sum_{v\in\bar{\opj}} \dist(v,{\opj}) \ge
\distintra(\op).
\]
\hfill\end{proof}

\subsection{Proof for Lemma \ref{lem:gpa:1cls:subm}}
\begin{proof}
Same as Theorem \ref{thm:intra-subm}, we prove this argument by finding a same mapping between a maximum matching $\{\ops{\edge}\}$ in the optimal solution with selected pairs $\{\hat{\edge}\}$.
With such a mapping, it is easy to see that for each $\hat{\edge}_i$,
\begin{itemize}
	\item If $\hat{\edge}_i \cap \ops{\edge}_{\pr{i}} = \emptyset$ or $\hat{\edge}_i$ is mapped to by only one $\ops{\edge}_{\pr{i}}$, $2  \feva_{\hat{\edge}_i}(\hat{\setofedge}^{i-1}) \ge \feva_{\ops{\edge}_{\pr{i}}}(\hat{\setofedge}^{i-1})$.
	\item If $\hat{\edge}_i$ is mapped to two   $\ops{\edge}_{\pr{i}},\ops{\edge}_{\ppr{i}}$, 
	$3  \feva_{\hat{\edge}_i}(\hat{\setofedge}^{i-1}) \ge \feva_{\ops{\edge}_{\pr{i}}}(\hat{\setofedge}^{i-1}) + \feva_{\ops{\edge}_{\ppr{i}}}(\hat{\setofedge}^{i-1})$. 
	This holds because one of $\ops{\edge}_{\pr{i}},\ops{\edge}_{\ppr{i}}$ must contain the first endpoint of $\hat{\edge}_i$ and have a value less than $\feva_{\hat{\edge}_i}(\hat{\setofedge}^{i-1})$.
\end{itemize}
Therefore we have
\begin{align*}
&3  \sum_{i=1}^{\nsel/2} \phi_{\hat{\edge}_i}(\hat{\setofedge}^{i-1})
\ge \sum_{\pr{i}=1}^{\nsel/2} \phi_{\ops{\edge}_{\pr{i}}}(\hat{\setofedge})\\
&= \sum_{\pr{i}=1}^{\nsel/2} Q(\hat{\setofedge}\cup\ops{\edge}_{\pr{i}}) - Q(\hat{\setofedge}) + \lambda 2(\nsel-1)\dist(\ops{\edge}_{\pr{i}})\\
&\ge Q(\hat{\setofedge}\cup\op) - Q(\hat{\setofedge}) + \sum_{\pr{i}=1}^{\nsel/2} \lambda 2(\nsel-1)\dist(\ops{\edge}_{\pr{i}})\\
&\ge \distintrasubm(\op) - Q(\hat{\setofedge})
\end{align*}
Finally,
\begin{align*}
3Q(\hat{\setofedge}) + \lambda \sum_{i=1}^{\nsel/2} 6 (\nsel-1) \dist(\hat{\edge}_{i}) &\ge \distintrasubm(\op) - Q(\hat{\setofedge})\\
6  \distintrasubm(\hat{\setofedge}) &\ge \distintrasubm(\op)
\end{align*}
\hfill\end{proof}

\subsection{Proof for Theorem \ref{thm:gpa:multi:subm}}
Same as Theorem \ref{thm:gpa:multi}, there exists an enhanced version.
The enhanced selection scheme is the same as its modular counterpart, but with a different objective.

\begin{proof}
The proof is a combination of Theorem \ref{thm:gpa:multi} and Lemma \ref{lem:gpa:1cls:subm}.
Following the same logic, we start from 
\begin{align*}
5 \sum_{i=1}^{\nselsum/2} \phi_{\hat{\edge_i}}(\hat{\setofedge}^{i-1})
\ge \sum_{j=1}^{\ncls}\sum_{\pr{i}=1}^{\nsel_j/2} \phi_{\ops{\edge}_{j\pr{i}}}(\hat{\setofedge})
\end{align*}
that eventually leads to 
\begin{align*}
10  \distintrasubm(\hat{\setofedge}) \ge \distintrasubm(\op)
\end{align*}
\hfill\end{proof}

\section{A general partition matroid constraint}
\label{subsec:partition}

We mention in Section \ref{sec:def} that our methods generalize to a general partition matroid constraint 
by removing a whole set of the partition when we select a single element from it.
Now we elaborate on this.

Obviously, this operation will give a valid solution, 
since we are not allowed to take more than one element from a partition.
As for theoretical guarantees, a central idea among the proofs for our algorithm 
is finding a mapping between pairs $\{\ops{\edge}\}$ in the optimal solution and pairs $\{\edge\}$ 
selected by our algorithm such that each $\edge$ is better than the $\ops{\edge}$ 
mapped to it due to the greedy nature.
The same logic applies to a general partition matroid constraint 
except for the definition of intersection between $\ops{\edge}$ and $\edge$.
Under a general constraint, a new definition is needed; a pair $\ops{\edge}$ 
is said to intersect with $\edge$ as long as one of their elements falls into a common partition.
Therefore, all of our analysis generalize smoothly to a general partition matroid constraint.


\begin{table*}[t]
	\footnotesize
	\centering
	\caption{Sensitivity of \gpa to the hyperparameter $\param$}
	\begin{tabular}{HlrrHrrrrr}
		\toprule
		{} & Dataset &      n &    b &  dim &  $\param$=0.1 &  $\param$=0.3 &  $\param$=0.5 &  $\param$=0.7 &  $\param$=0.95 \\
		exp\_id &         &        &      &    &            &            &            &            &             \\
		\midrule
		427    &  random &   1000 &   10 &  2 &      0.996 &      0.994 &      0.994 &      0.995 &       1.0 \\
		428    &  random &   1000 &  100 &  2 &      0.997 &      0.973 &      0.962 &      0.962 &       1.0 \\
		429    &  random &  10000 &   10 &  2 &      0.994 &      0.991 &      0.997 &      0.997 &       1.0 \\
		430    &  random &  10000 &  100 &  2 &      0.998 &      0.998 &      1.0 &      0.998 &       0.999 \\
		\midrule
		423    &   proto &   1000 &   10 &  2 &      0.992 &      0.995 &      0.998 &      0.998 &       1.0 \\
		424    &   proto &   1000 &  100 &  2 &      0.948 &      0.998 &      0.994 &      0.987 &       1.0 \\
		425    &   proto &  10000 &   10 &  2 &      0.990 &      0.996 &      0.993 &      1.0 &       0.996 \\
		426    &   proto &  10000 &  100 &  2 &      0.999 &      0.999 &      0.999 &      1.0 &       1.0 \\
		\bottomrule
	\end{tabular}
	\label{tbl:param}
\end{table*}

\begin{table*}[t]
	\footnotesize
	\centering
	\caption{Performance for large synthetic datasets ($n=100000$, 10 clusters) where $dim$ stands for dimension.}
	\begin{tabular}{HlHrrrrrrrrr}
		\toprule
		{} & Dataset &       n &    b &   dim &  GPmin &  GPavg &  GPmax &  GVmin &  GVavg &  GVmax &     RN \\
		exp\_id &         &         &      &     &        &        &        &        &        &        &        \\
		\midrule
		481    &  random &  100000 &   10 &   2 &  0.980 &  0.992 &    1.0 &  0.983 &  0.987 &  0.993 &  0.576 \\
		482    &  random &  100000 &   10 &  10 &  0.987 &  0.995 &    1.0 &  0.985 &  0.989 &  0.994 &  0.678 \\
		483    &  random &  100000 &  100 &   2 &  0.998 &  0.999 &    1.0 &  0.997 &  0.998 &  0.998 &  0.629 \\
		497    &  random &  100000 &  100 &  10 &  0.999 &  1.000 &    1.0 &  0.998 &  0.998 &  0.998 &  0.745 \\
		\midrule
		477    &   proto &  100000 &   10 &   2 &  0.985 &  0.993 &    1.0 &  0.986 &  0.988 &  0.990 &  0.170 \\
		478    &   proto &  100000 &   10 &  10 &  0.989 &  0.994 &    1.0 &  0.985 &  0.990 &  0.996 &  0.549 \\
		479    &   proto &  100000 &  100 &   2 &  0.999 &  0.999 &    1.0 &  0.995 &  0.997 &  0.997 &  0.225 \\
		484    &   proto &  100000 &  100 &  10 &  0.999 &  0.999 &    1.0 &  0.997 &  0.998 &  0.998 &  0.600 \\
		\bottomrule
	\end{tabular}
	\label{tbl:syn}
\end{table*}

\begin{figure} 
	\centering
	\includegraphics[width=.45\textwidth]{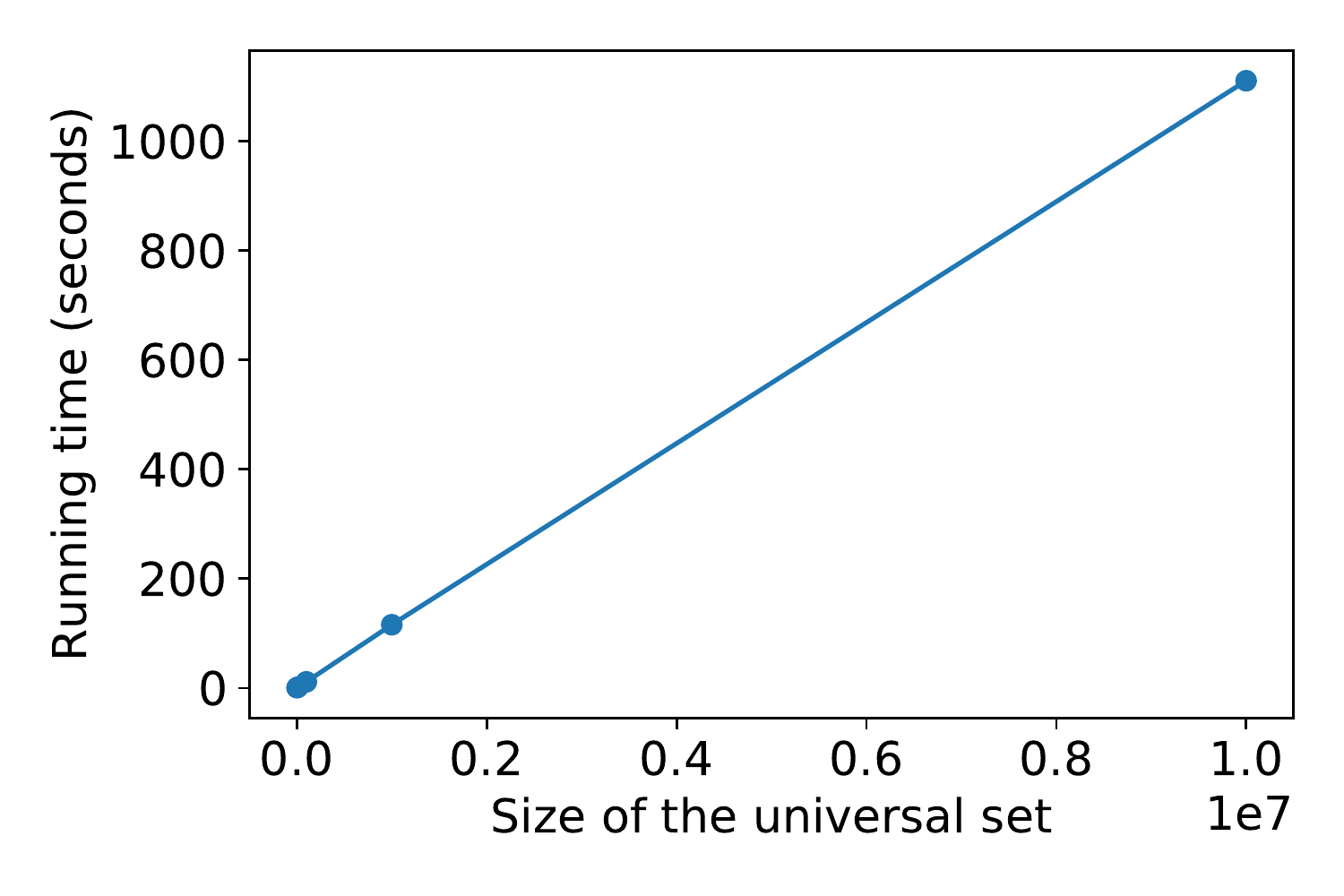} 
	\caption{Linearly scalability of the \gpa algorithm on the random dataset with 10 clusters, a fixed $\nsel$ and $\param=0.95$.} \label{fig:time}
\end{figure}

\begin{table*}[t]
	\footnotesize
	\centering
	\caption{Performance on documents aggregations, where \textit{th} stands for the topic probability threshold.}
	\begin{tabular}{Hlrrrrrrrrrr}
		\toprule
		{} & Dataset &   b &  m &    th &  GPmin &  GPavg &  GPmax &  GVmin &  GVavg &  GVmax &     RN \\
		exp\_id &         &     &       &       &        &        &        &        &        &        &        \\
		\midrule
		485    &  20news &  10 &     5 &  0.15 &  0.970 &  0.985 &    1.0 &  0.967 &  0.977 &  0.990 &  0.543 \\
		486    &  20news &  10 &     5 &  0.25 &  0.987 &  0.996 &    1.0 &  0.976 &  0.981 &  0.991 &  0.525 \\
		487    &  20news &  30 &     5 &  0.15 &  0.988 &  0.993 &    1.0 &  0.986 &  0.989 &  0.991 &  0.574 \\
		488    &  20news &  30 &     5 &  0.25 &  0.997 &  0.998 &    1.0 &  0.991 &  0.995 &  0.998 &  0.559 \\
		\midrule
		489    &    blog &  10 &     8 &  0.15 &  0.997 &  0.999 &    1.0 &  0.988 &  0.992 &  0.995 &  0.442 \\
		490    &    blog &  10 &     8 &  0.25 &  0.988 &  0.995 &    1.0 &  0.988 &  0.991 &  0.997 &  0.433 \\
		491    &    blog &  30 &     8 &  0.15 &  0.998 &  0.999 &    1.0 &  0.996 &  0.997 &  0.998 &  0.465 \\
		496    &    blog &  30 &     8 &  0.25 &  0.999 &  0.999 &    1.0 &  0.995 &  0.997 &  0.998 &  0.440 \\
		\bottomrule
	\end{tabular}
	\label{tbl:doc}
\end{table*}

\begin{table}[t]
\footnotesize
\centering
\caption{A use case for descriptors of seven data mining communities.}
\begin{tabular}{ll}
\toprule
C1 
&  database managementsystem r \\
&  previous paper \\
&  xml data model\\
&  shared data\\
&  relational model\\
&  data base system\\ 
\midrule
C2
&  data compression\\
&  big data\\
&  bandwidth\\
&  efficient data restructuring\\
&  aggregating data access\\
&  data point\\ 
\midrule
C3
&  data store\\
&  data processing\\
&  data instance\\
&  health data\\
&  large system\\
&  database machine\\ 
\midrule
C4 
&  semantic data\\
&  content management middleware\\
&  spatial data\\
&  census data\\
&  information granularity\\
&  associates data output\\ 
\midrule
C5 
&  2-dimensional data visualization\\
&  archival data\\
&  high dimensional data\\
&  complex scientific data\\
&  italian workshop\\
&  data distribution\\ 
\midrule
C6
&  data element\\
&  sequence data\\
&  structured data\\
&  data node\\
&  data cube\\
&  game engine\\ 
\midrule
C7
&  feature model\\
&  data warehousing\\
&  complex object\\
&  temporal database management system\\
&  abstract data type\\
&  historical database\\   
\bottomrule
\end{tabular}
\label{tbl:usecase}
\end{table}

\end{appendices}
\flushend

\end{document}